\documentclass{mgiartcl}
\pdfoutput=1

\usepackage{tikz}
\usetikzlibrary{arrows}
\usetikzlibrary{shapes}
\tikzset{every picture/.style={>=stealth,bend angle=20}}
\tikzset{every label/.style={font=\scriptsize}}
\tikzset{every node/.style={font=\scriptsize}}
\tikzset{play/.style={circle,draw,minimum size=#1}}
\tikzset{play/.default=0.625cm}
\tikzset{prob/.style={diamond,draw,minimum size=#1}}
\tikzset{prob/.default=0.7cm}
\tikzset{end/.style={rectangle,draw,minimum size=#1}}
\tikzset{end/.default=0.55cm}
\usepackage{booktabs}

\title{Decision Problems for Nash Equilibria \\ in Stochastic Games\thanks{This
work was supported by the DFG Research Training Group~1298
(\textsc{AlgoSyn}) and the ESF Research Networking Programme \textsc{Games}.}}
\author{Michael Ummels\addr{1} and Dominik Wojtczak\addr{2}}
\address{RWTH Aachen University, Germany \\
\email{ummels@logic.rwth-aachen.de}}
\address{CWI, Amsterdam, The Netherlands \\
\email{d.k.wojtczak@cwi.nl}}

\DeclareMathOperator{\id}{id}

\DeclareMathOperator{\Prob}{Pr}

\DeclareMathOperator{\val}{val}
\newcommand{\Win}{\mbox{Win}}
\DeclareMathOperator{\Reach}{Reach}
\DeclareMathOperator{\Inf}{Inf}
\DeclareMathOperator{\StreettEC}{StreettEC}
\DeclareMathOperator{\RabinEC}{RabinEC}

\newcommand{\NE}{\mbox{NE}\xspace}
\newcommand{\QualNE}{\mbox{QualNE}\xspace}
\newcommand{\PureNE}{\mbox{PureNE}\xspace}
\newcommand{\StatNE}{\mbox{StatNE}\xspace}
\newcommand{\PosNE}{\mbox{PosNE}\xspace}
\newcommand{\FinNE}{\mbox{FinNE}\xspace}
\newcommand{\SqrtSum}{\mbox{SqrtSum}\xspace}
\newcommand{\PPAD}{\mbox{PPAD}\xspace}
\newcommand{\ExThR}{\ensuremath{\mathrm{ExTh}(\mathfrak{R})}\xspace}

\newcommand{\SOneS}{\textrm{S1S}\xspace}

\newcommand{\cf}{cf.\xspace}
\newcommand{\ea}{et al.\xspace}
\newcommand{\eg}{e.g.\xspace}
\newcommand{\ie}{i.e.\xspace}

\newcommand{\wlg}{w.l.o.g.\xspace}

\newif\ifapdx
\apdxtrue

\begin{document}
\maketitle
\setcounter{footnote}{0}

\begin{abstract}
We analyse the computational complexity of finding Nash equilibria in
stochastic multiplayer games with $\omega$-regular objectives.\linebreak
While the existence of an equilibrium whose payoff falls into a certain
interval may be undecidable, we single out several decidable
restrictions of the problem. First, restricting the search space
to stationary, or pure stationary, equilibria results in problems
that are typically contained in \PSpace and \NP, respectively. Second, we
show that the existence of an equilibrium with a binary payoff
(\ie an equilibrium where each player either wins or loses with
probability~1) is decidable. We also establish that the existence
of a Nash equilibrium with a certain binary payoff entails the existence
of an equilibrium with the same payoff in pure, finite-state strategies.
\end{abstract}

\section{Introduction}

We study \emph{stochastic games} \cite{NeymanS03} played by multiple
players on a finite, directed graph. Intuitively, a play of such a game evolves
by moving a token along edges of the graph: Each vertex of the graph is either
controlled by one of the players, or it is \emph{stochastic}. Whenever
the token arrives at a non-stochastic vertex, the player who controls this
vertex must move the token to a successor vertex; when the token arrives at a
stochastic vertex, a fixed probability distribution determines the next vertex.
A measurable function maps plays to payoffs.
In the simplest case, which we discuss here, the possible payoffs of a single
play are binary (\ie each player either wins or loses a given play). However,
due to the presence of stochastic vertices, a player's \emph{expected payoff}
(\ie her probability of winning) can be an arbitrary probability.

Stochastic games with $\omega$-regular objectives have been
successfully applied in the verification and
synthesis of reactive systems under the influence of random events. Such a 
system is usually modelled as a game between the
system and its environment, where the environment's objective is the complement
of the system's objective: the environment is considered hostile. Therefore,
the research in this area has traditionally focused on two-player games
where each play is won by precisely one of the two players, so-called
\emph{two-player zero-sum games}. However, the system may comprise of
several components with independent objectives, a situation which is naturally
modelled by a multiplayer game.

The most common interpretation of rational behaviour in multiplayer games is
captured by the notion of a \emph{Nash equilibrium} \cite{Nash50}. In a Nash
equilibrium, no player can improve her payoff by unilaterally switching to a
different strategy. Chatterjee~\ea \cite{ChatterjeeJM04} gave an algorithm for
computing a Nash equilibrium in a stochastic multiplayer games with
$\omega$-regular winning conditions. We argue that this is not
satisfactory. Indeed, it can be shown that their algorithm may compute an
equilibrium where all players lose almost surely (\ie receive expected
payoff~0), while there exist other equilibria where all players win almost
surely (\ie receive expected payoff~1).

In applications, one might look for an equilibrium where as many players as
possible win almost surely or where it is guaranteed that the expected payoff
of the equilibrium falls into a certain interval. Formulated as a
decision problem, we want to know, given a $k$-player game~$\calG$ with
initial vertex~$v_0$ and two thresholds $\vec{x},\vec{y}\in [0,1]^k$,
whether $(\calG,v_0)$ has a Nash equilibrium with expected payoff at
least~$\vec{x}$ and at most~$\vec{y}$. This problem, which we call \NE for
short, is a generalisation of the \emph{quantitative decision problem} for
two-player zero-sum games, which asks whether in such a game player~0 has
a strategy that ensures to win the game with a probability that lies above
a given threshold.

In this paper, we analyse the decidability of \NE for games with
$\omega$-regular objectives. Although the decidability of \NE remains
open, we can show that several restrictions of \NE are decidable: First, we
show that \NE becomes decidable when one restricts the search space to
equilibria in positional (\ie pure, stationary), or stationary, strategies,
and that the resulting decision problems typically lie in \NP and \PSpace,
respectively  (\eg if the objectives are specified as Muller
conditions). Second, we show that the following \emph{qualitative} version of
\NE is decidable: Given a $k$-player game~$\calG$ with initial vertex~$v_0$ and
a \emph{binary} payoff~$\vec{x}\in\{0,1\}^k$, decide whether $(\calG,v_0)$ has
a Nash equilibrium with expected payoff~$\vec{x}$. Moreover, we prove that,
depending on the representation of the objective, this problem is
typically complete for one of the complexity classes \PTime, \NP, \coNP and
\PSpace, and that the problem is invariant under restricting the
search space to equilibria in pure, finite-state strategies.

Our results have to be viewed in light of the (mostly) negative results
we derived in \cite{UmmelsW09}. In particular, it was shown in \cite{UmmelsW09}
that \NE becomes
undecidable if one restricts the search space to equilibria in pure strategies
(as opposed to equilibria in possibly mixed strategies), even for \emph{simple
stochastic multiplayer games}. These are games with simple reachability
objectives. The undecidability result crucially makes use of the
fact that the Nash equilibrium one is looking for can have a payoff that is not
binary. Hence, this result cannot be applied to the qualitative version of \NE,
which we show to be decidable in this paper. It was also proven in
\cite{UmmelsW09}
that the problems that arise from \NE when one restricts the search space to
equilibria in positional or stationary strategies are both \NP-hard. Moreover,
we showed that the restriction to stationary strategies is at least as hard as
the problem \SqrtSum \cite{AllenderBKM06}, a problem which is not known to
lie inside the polynomial hierarchy.
This demonstrates that the upper bounds we prove for these problems in this
paper will be hard to improve.

\subsubsection{Related Work}

Determining the complexity of Nash equilibria has attracted much
interest in recent years. In particular, a series of papers culminated in the
result that computing a Nash equilibrium of a two-player game in strategic form
is complete for the complexity class \PPAD \cite{DaskalakisGP06,ChenD06}.
However, the work closest to ours is \cite{Ummels08}, where the
decidability of (a variant of) the qualitative version of \NE in infinite games
without stochastic vertices was proven. Our results complement the results in
that paper,
and although our decidability proof for the qualitative setting is structurally
similar to the one in \cite{Ummels08}, the presence of stochastic
vertices makes the proof substantially more challenging.

Another subject that is related to the study of stochastic multiplayer games
are \emph{Markov decision processes with multiple objectives}. These can be
viewed as stochastic multiplayer games where all non-stochastic vertices are
controlled by a single player. For $\omega$-regular objectives,
Etessami \ea \cite{EtessamiKVY08} proved the decidability of \NE for these
games. Due to the different nature of the restrictions, this result is
incomparable to our results.

\section{Preliminaries}\label{sect:prelim}

The model of a \emph{(two-player zero-sum) stochastic game}
\cite{Condon92} easily generalises to the multiplayer case: Formally,
a \emph{stochastic multiplayer game (SMG)} is a tuple
$\calG=(\Pi,V,(V_i)_{i\in\Pi},\Delta,(\Win_i)_{i\in\Pi})$
where
\begin{itemize}
 \item $\Pi$ is a finite set of \emph{players} (usually
$\Pi=\{0,1,\dots,k-1\}$);
 \item $V$ is a finite, non-empty set of \emph{vertices};
 \item $V_i\subseteq V$ and $V_i\cap V_j=\emptyset$ for each $i\not=j\in\Pi$;
 \item $\Delta\subseteq V\times([0,1]\cup\{\bot\})\times V$ is the
  \emph{transition relation};
 \item $\Win_i\subseteq V^\omega$ is a Borel set for each $i\in\Pi$.
\end{itemize}
The structure $G=(V,(V_i)_{i\in\Pi},\Delta)$ is called the \emph{arena} of
$\calG$, and $\Win_i$ is called the \emph{objective}, or the
\emph{winning condition}, of
player~$i\in\Pi$. A vertex $v\in V$ is \emph{controlled by player~$i$} if
$v\in V_i$ and a \emph{stochastic vertex} if $v\not\in\bigcup_{i\in\Pi} V_i$.

We require that a transition is labelled by a probability iff
it originates in a stochastic vertex: If $(v,p,w)\in\Delta$ then $p\in
[0,1]$ if $v$ is a stochastic vertex and $p=\bot$ if $v\in V_i$ for some $i\in
\Pi$. Additionally, for each pair of a stochastic vertex~$v$ and an arbitrary
vertex~$w$, we require that there exists precisely one $p\in[0,1]$ such that
$(v,p,w)\in\Delta$.  Moreover, for each stochastic vertex $v$, the outgoing
probabilities must sum up to 1: $\sum_{(p,w):(v,p,w)\in \Delta} p=1$.
Finally, we require that for each vertex the set $v\Delta:=\{w\in V:
\text{exists $p\in(0,1]\cup\{\bot\}$ with $(v,p,w)\in\Delta$}\}$ is non-empty,
\ie every vertex has at least one successor.

A special class of SMGs are \emph{two-player zero-sum stochastic games
(2SGs)}. These are SMGs played by only two players (player~0 and player~1) and
one player's objective is the complement of the other player's objective, \ie
$\Win_0=V^\omega\setminus\Win_1$. An even more restricted
model are \emph{one-player stochastic games}, also known as \emph{Markov
decision processes (MDPs)}, where there is only one player (player~0). Finally,
\emph{Markov chains} are SMGs with no players at all, \ie there are only
stochastic vertices.


\subsubsection{Strategies and strategy profiles}

In the following, let $\calG$ be an arbitrary SMG.
A \emph{(mixed) strategy of player~$i$ in~$\calG$} is a mapping
$\sigma:V^*V_i\to\calD(V)$ assigning to each possible \emph{history}
$xv\in V^* V_i$ of vertices ending in a vertex controlled by player~$i$ a
(discrete) probability distribution over $V$ such that $\sigma(xv)(w)>0$ only
if $(v,\bot,w)\in\Delta$. Instead of $\sigma(xv)(w)$, we usually write
$\sigma(w\mid xv)$.
A \emph{(mixed) strategy profile of $\calG$} is a tuple
$\vec{\sigma}=(\sigma_i)_{i\in\Pi}$ where $\sigma_i$ is a strategy of
player~$i$ in~$\calG$.
Given a strategy profile $\vec{\sigma}=(\sigma_j)_{j\in\Pi}$ and a
strategy~$\tau$ of player~$i$, we denote by $(\vec{\sigma}_{-i},
\tau)$ the strategy profile resulting from $\vec{\sigma}$ by replacing
$\sigma_i$ with $\tau$.

A strategy $\sigma$ of player~$i$ is called \emph{pure}
if for each $xv\in V^*V_i$ there exists $w\in v\Delta$ with $\sigma(w\mid xv)
=1$. Note that a pure strategy of player~$i$ can be identified with a function
$\sigma:V^*V_i \to V$. A strategy profile $\vec{\sigma}=(\sigma_i)_{i\in\Pi}$
is called \emph{pure} if each $\sigma_i$ is pure.

A strategy $\sigma$ of player~$i$ in $\calG$ is called \emph{stationary}
if $\sigma$ depends only on the current vertex: $\sigma(xv)=\sigma(v)$
for all $xv\in V^*V_i$. Hence, a stationary strategy of player~$i$ can be
identified with a function $\sigma:V_i\to\calD(V)$.
A strategy profile $\vec{\sigma}=(\sigma_i)_{i\in\Pi}$ of $\calG$ is called
\emph{stationary} if each $\sigma_i$ is stationary.

We call a pure, stationary strategy a \emph{positional strategy} and
a strategy profile consisting of positional strategies only a
\emph{positional strategy profile}. Clearly, a positional strategy of
player~$i$ can be identified with a function $\sigma:V_i\to V$.
More generally, a pure strategy $\sigma$ is called \emph{finite-state} if it
can be implemented by a finite
automaton with output or, equivalently, if the equivalence relation
$\mathord{\sim}\subseteq V^*\times V^*$ defined by $x\sim y$ if $\sigma(xz)=
\sigma(yz)$ for all $z\in V^*V_i$ has only finitely many equivalence
classes.
Finally, a \emph{finite-state strategy profile} is a profile consisting of
finite-state strategies only.

It is sometimes convenient to designate an initial vertex $v_0\in V$ of the
game. We call the tuple $(\calG,v_0)$ an \emph{initialised SMG}.
A strategy (strategy profile)
of $(\calG,v_0)$ is just a strategy (strategy profile) of $\calG$.
In the following, we will use the abbreviation SMG also for initialised SMGs.
It should always be clear from the context if the game is initialised or not.

Given an initial vertex~$v_0$ and a strategy profile $\vec{\sigma}=
(\sigma_i)_{i\in\Pi}$, the \emph{conditional probability of $w\in V$ given the
history $xv\in V^*V$} is the number $\sigma_i(w\mid xv)$ if $v\in V_i$ and the
unique $p\in [0,1]$ such that $(v,p,w)\in\Delta$ if $v$ is a stochastic vertex.
We abuse notation and denote this probability by $\vec{\sigma}(w\mid xv)$. The
probabilities $\vec{\sigma}(w\mid xv)$ induce a probability measure on the
space $V^\omega$ in the following way: The probability of a basic open set
$v_1\dots v_k\cdot V^\omega$ is 0 if $v_1\not=v_0$ and the product of the
probabilities $\vec{\sigma}(v_j\mid v_1\dots v_{j-1})$ for $j=2,\dots,k$
otherwise. It is a classical result of measure theory that this extends to a
unique probability measure assigning a probability to every Borel subset of
$V^\omega$, which we denote by $\Prob_{v_0}^{\vec{\sigma}}$.

Given a strategy $\sigma$ and a sequence $x\in V^*$, we define the
\emph{residual strategy} $\sigma[x]$ by $\sigma[x](yv)=\sigma(xyv)$.
If $\vec{\sigma}=(\sigma_i)_{i\in\Pi}$ is a strategy profile, then the
\emph{residual strategy profile} $\vec{\sigma}[x]$ is just the profile
of the residual strategies $\sigma_i[x]$. The following two lemmas are
taken from \cite{Zielonka04}.

\begin{lemma}\label{lemma:equivalent-profiles}
Let~$\vec{\sigma}$ and $\vec{\tau}$ be two strategy profiles of~$\calG$,
equal over a prefix-closed set $X\subseteq V^*$. Then 
$\Prob^{\vec{\sigma}}_{v_0}(B)=\Prob^{\vec{\tau}}_{v_0}(B)$ for every
Borel set~$B$ all of whose prefixes belong to~$X$.
\end{lemma}

\begin{lemma}\label{lemma:residual-strategies}
Let $\vec{\sigma}$ be any strategy profile of~$\calG$, $xv\in V^*V$ a
history of~$\calG$, and $B\subseteq V^\omega$ a Borel set. Then
$\Prob^{\vec{\sigma}}_{v_0}(B\cap xv\cdot V^\omega)=\Prob^{\vec{\sigma}}_{v_0}
(xv\cdot V^\omega)\cdot \Prob^{\vec{\sigma}[x]}_{v}(B[x])$,
where $B[x]:=\{\alpha\in V^\omega: x\alpha\in B\}$.
\end{lemma}

For a strategy profile~$\vec{\sigma}$, we are mainly interested in the
probabilities $p_i:=\Prob^{\vec{\sigma}}_{v_0}(\Win_i)$ of winning. We call
$p_i$ the \emph{(expected) payoff of $\vec{\sigma}$ for player~$i$} and the
vector $(p_i)_{i\in\Pi}$ the \emph{(expected) payoff of $\vec{\sigma}$}.

\subsubsection{Subarenas and end components}

Given an SMG~$\calG$, we call a set $U\subseteq V$ a \emph{subarena of $\calG$}
if 1.\ $U\not=\emptyset$; 2.\ $v\Delta\cap U\not=\emptyset$ for
each $v\in U$, and 3.\ $v\Delta\subseteq U$ for each stochastic
vertex~$v\in U$.

A set $C\subseteq V$ is called an \emph{end component of~$\calG$} if $C$ is a
subarena, and additionally $C$ is strongly connected: for every pair of
vertices $v,w\in C$ there exists a sequence $v=v_1,v_2,\dots,v_n=w$ with
$v_{i+1}\in v_i\Delta$ for each $0<i<n$. An end component~$C$ is \emph{maximal
in a set $U\subseteq V$} if there is no end component~$C'\subseteq U$ with
$C\subsetneq C'$. For any subset $U\subseteq V$, the set of all end components
maximal in $U$ can be computed by standard graph algorithms in quadratic time
(see \eg \cite{Alfaro97}).

The central fact about end components is that, under \emph{any} strategy
profile, the set of vertices visited infinitely often is almost surely an
end component. For an infinite sequence $\alpha$, we denote by $\Inf(\alpha)$
the set of elements occurring infinitely often in $\alpha$.

\begin{lemma}[\cite{Alfaro97,CourcoubetisY95}]\label{lemma:end-components}
Let $\calG$ be any SMG, and let $\vec{\sigma}$ be any strategy profile
of~$\calG$. Then $\Prob^{\vec{\sigma}}_v(\{\alpha\in V^\omega:
\text{$\Inf(\alpha)$ is an end component}\})=1$ for each vertex~$v\in V$.
\end{lemma}

Moreover, for any end component~$C$, we can construct a
stationary strategy profile~$\vec{\sigma}$ that, when started in~$C$, 
guarantees to visit all (and only) vertices in~$C$ infinitely often.

\begin{lemma}[\cite{Alfaro97,CourcoubetisY98}]\label{lemma:end-component-strategy}
Let $\calG$ be any SMG, and let $C$ be any end component of~$\calG$. There
exists a stationary strategy profile~$\vec{\sigma}$ with
$\Prob^{\vec{\sigma}}_v(\{\alpha\in V^\omega:\Inf(\alpha)=C\})=1$ for
each vertex~$v\in C$.
\end{lemma}

\subsubsection{Values, determinacy and optimal strategies}

Given a strategy~$\tau$ of player~$i$ in $\calG$ and a vertex~$v\in V$, the
\emph{value of $\tau$ from $v$} is the number $\val^\tau(v):=
\inf_{\vec{\sigma}}\Prob^{\vec{\sigma}_{-i},\tau}_{v}(\Win_i)$, where
$\vec{\sigma}$ ranges over all strategy profiles of~$\calG$. Moreover,
we define the \emph{value of~$\calG$ for player~$i$ from~$v$} as the supremum
of these values, \ie $\val_i^\calG(v)=\sup_\tau \val^\tau(v)$, where $\tau$
ranges over all strategies of player~$i$ in~$\calG$.
Intuitively, $\val_i^\calG(v)$ is the maximal payoff that player~$i$ can
ensure when the game starts from~$v$.
If $\calG$ is a two-player zero-sum game, a celebrated theorem due to Martin
\cite{Martin98} states that the game is \emph{determined}, \ie $\val_0^\calG=
1-\val_1^\calG$ (where the equality holds pointwise). The number
$\val^\calG(v):=\val_0^\calG(v)$ is consequently called the \emph{value
of~$\calG$ from~$v$}.

Given an initial vertex $v_0\in V$, a strategy~$\sigma$ of player~$i$ in $\calG$
is called \emph{optimal} if $\val^\sigma(v_0)=\val_i^\calG(v_0)$. A
\emph{globally optimal} strategy is a strategy that is optimal for every
possible initial vertex $v_0\in V$.
Note that optimal strategies do not need to exist since the supremum in the
definition of $\val_i^\calG$ is not necessarily attained. However, if for
every possible initial vertex there exists an optimal strategy, then there
also exists a globally optimal strategy.

\subsubsection{Objectives}

We have introduced objectives as abstract Borel sets of infinite
sequences of vertices; to be amendable for algorithmic solutions, all objectives
must be finitely representable. In verification, objectives are usually
\emph{$\omega$-regular sets} specified by formulae of the logic \SOneS (monadic
second-order logic on infinite words) or \LTL (linear-time temporal logic)
referring to unary predicates $P_c$ indexed by a finite set $C$ of
\emph{colours}. These are interpreted as winning conditions in a game by
considering a \emph{colouring} $\chi:V\to C$ of the vertices in the game.
Special cases are the following well-studied conditions:
\begin{itemize}
\item\emph{B\"uchi} (given by a  set $F\subseteq C$): the set
of all $\alpha \in C^\omega$ such that $\Inf(\alpha)\cap F\not=\emptyset$.

\item\emph{co-B\"uchi} (given by set $F\subseteq C$): the
set of all $\alpha\in C^\omega$ such that $\Inf(\alpha)\subseteq F$.

\item\emph{Parity} (given by a \emph{priority function} $\Omega:C\to\bbN$):
the set of all $\alpha\in C^\omega$ such that $\min(\Inf(\Omega(\alpha)))$ is
even.

\item\emph{Streett} (given by a set $\Omega$ of pairs $(F,G)$
where $F,G\subseteq C$): the set of all $\alpha\in C^\omega$ such that
for all pairs $(F,G)\in\Omega$ with $\Inf(\alpha)\cap F\not=\emptyset$ it is
the case that $\Inf(\alpha)\cap G\not=\emptyset$.

\item\emph{Rabin} (given by a set $\Omega$ of pairs $(F,G)$ where
$F,G\subseteq C$): the set of all $\alpha\in C^\omega$ such that
there exists a pair $(F,G)\in\Omega$ with $\Inf(\alpha)\cap F\not=\emptyset$
but $\Inf(\alpha)\cap G=\emptyset$.

\item\emph{Muller} (given by a family $\calF$ of sets
$F\subseteq C$): the set of all $\alpha\in C^\omega$ such that there
exists $F\in\calF$ with $\Inf(\alpha)=F$.
\end{itemize}
Note that any B\"uchi condition is a parity condition with
two priorities, that any parity condition is both a Streett and a Rabin
condition, and that any Streett or Rabin condition is a Muller condition.
(However, the translation from a set of Streett/Rabin pairs to an equivalent
family of accepting sets is, in general, exponential.)
In fact, the intersection (union) of any two parity conditions
is a Streett (Rabin) condition. Moreover, the complement of a B\"uchi (Streett)
condition is a co-B\"uchi (Rabin) condition and vice versa, whereas the class
of parity conditions and the class of Muller conditions are closed under
complementation. Finally, note that any of the above condition
is \emph{prefix-independent}: for every $\alpha\in C^\omega$ and $x\in
C^*$, $\alpha$ satisfies the condition iff $x\alpha$ does.

Theoretically, parity and Rabin conditions provide the best balance of
expressiveness and simplicity: On the one hand, any SMG where player~$i$ has a
Rabin objective admits a globally optimal positional strategy for this
player \cite{ChatterjeeAH05}. On the other hand,
any SMG with $\omega$-regular objectives can be reduced to an SMG
with parity objectives using \emph{finite memory} (see \cite{Thomas95}).
An important consequence of this reduction is that there exist
globally optimal finite-state strategies in every SMG with $\omega$-regular
objectives. In fact, there exist globally optimal pure strategies in every
SMG with prefix-independent objectives \cite{HornG08}.

In the following, for the sake of simplicity, we will only consider games where
each vertex is coloured by itself, \ie $C=V$ and $\chi=\id$. We would like to
point out, however, that all our results remain valid for games with other
colourings. For the same reason, we will usually not distinguish between a
condition and its finite representation.

\subsubsection{Decision problems for two-player zero-sum games}

The main computational problem for two-player zero-sum games is computing
the value (and optimal strategies for either player, if they exist).
Rephrased as a decision problem, the problem looks as follows:
\begin{quote}
Given a 2SG $\calG$, an initial vertex~$v_0$ and a
rational probability~$p$, decide whether $\val^\calG(v_0)\geq p$.
\end{quote}
A special case of this problem arises for $p=1$. Here, we only want to know
whether player~0 can win the game almost surely (in the limit). Let us
call the former problem the \emph{quantitative} and the latter problem the
\emph{qualitative decision problem for 2SGs}.

\Cref{table:summary-results} summarises the results about the complexity
of the quantitative and the qualitative decision problem for two-player
zero-sum stochastic games depending on the type of player~0's objective.
For MDPs, both problems are decidable in polynomial time for
all of the aforementioned objectives (\ie up to Muller conditions)
\cite{Chatterjee07,Alfaro97}.
\begin{table}
\begin{tabular}{rcc}
\toprule
& Quantitative & Qualitative \\
\midrule
(co-)B\"uchi & $\NP\cap\coNP$ \cite{ChatterjeeJH04}
 & \PTime-complete \cite{AlfaroH00} \\
Parity & $\NP\cap\coNP$ \cite{ChatterjeeJH04}
 & $\NP\cap\coNP$ \cite{ChatterjeeJH04} \\
Streett & \coNP-complete\ \cite{ChatterjeeAH05,EmersonJ88}
 & \coNP-complete\ \cite{ChatterjeeAH05,EmersonJ88} \\
Rabin & \NP-complete\ \cite{ChatterjeeAH05,EmersonJ88}
 & \NP-complete\ \cite{ChatterjeeAH05,EmersonJ88} \\
Muller & \PSpace-complete\ \cite{Chatterjee07,HunterD05}
 & \PSpace-complete\ \cite{Chatterjee07,HunterD05} \\
\bottomrule
\end{tabular}
\caption{\label{table:summary-results}The complexity of deciding the value in
2SGs.}
\end{table}

\section{Nash equilibria and their decision problems}

To capture rational behaviour of (selfish) players, John Nash \cite{Nash50}
introduced the notion of, what is now called, a \emph{Nash equilibrium}.
Formally, given a strategy profile~$\vec{\sigma}$ in an SMG $(\calG,v_0)$, a
strategy~$\tau$ of player~$i$ is called a \emph{best response to
$\vec{\sigma}$} if $\tau$ maximises the expected payoff of player~$i$:
$\Prob_{v_0}^{\vec{\sigma}_{-i},\tau'}(\Win_i)\leq
\Prob_{v_0}^{\vec{\sigma}_{-i},\tau}(\Win_i)$ for all strategies $\tau'$
of player~$i$. A Nash equilibrium is a strategy profile $\vec{\sigma}=
(\sigma_i)_{i\in\Pi}$ such that each $\sigma_i$ is a best response to
$\vec{\sigma}$. Hence, in a Nash equilibrium no player can improve her payoff
by (unilaterally) switching to a different strategy. For two-player zero-sum
games, a Nash equilibrium is nothing else than a pair of optimal strategies.

\begin{proposition}\label{prop:two-zerosum-nash}
Let $(\calG,v_0)$ be a two-player zero-sum game. A strategy profile
$(\sigma,\tau)$ of $(\calG,v_0)$ is a Nash equilibrium iff both $\sigma$
and $\tau$ are optimal. In particular, every Nash equilibrium of $(\calG,v_0)$
has payoff~$(\val^\calG(v_0),1-\val^\calG(v_0))$.
\end{proposition}
\begin{proof}
$(\Rightarrow)$ Assume that both $\sigma$ and $\tau$ are optimal, but
that $(\sigma,\tau)$ is not a Nash equilibrium. Hence, one of the players,
say player~1, can improve her payoff by playing some strategy $\tau'$. Hence,
$\val^\calG(v_0)=\Prob^{\sigma,\tau}_{v_0}(\Win_0)>
\Prob^{\sigma,\tau'}_{v_0}(\Win_0)$.
However, since $\sigma$ is optimal, it must also be the case that
$\val^\calG(v_0)\leq\Prob^{\sigma,\tau'}_{v_0}(\Win_0)$, a contradiction.
The reasoning in the case that player~0 can improve is analogous.

$(\Leftarrow)$ Let $(\sigma,\tau)$ be a Nash equilibrium of $(\calG,v_0)$,
and let us first assume that $\sigma$ is not optimal, \ie
$\val^\sigma(v_0)<\val^\calG(v_0)$. By the definition of $\val^\calG$,
there exists another strategy $\sigma'$ of player~0
such that
$\val^\sigma(v_0)<\val^{\sigma'}(v_0)\leq\val^\calG(v_0)$.
Moreover, since $(\sigma,\tau)$ is a Nash equilibrium:
\[
\Prob^{\sigma,\tau}_{v_0}(\Win_0)
\leq\val^\sigma(v_0)
<\val^{\sigma'}(v_0)=
\inf\nolimits_\tau\Prob^{\sigma',\tau}_{v_0}(\Win_0)
\leq \Prob^{\sigma',\tau}_{v_0}(\Win_0)\,.
\]
Thus player~0 can improve her payoff by playing $\sigma'$ instead of
$\sigma$, a contradiction to the fact that $(\sigma,\tau)$ is a Nash
equilibrium. Now, if we assume that $\tau$ is not optimal, we can
analogously show the existence of a strategy~$\tau'$ that player~1
can use to improve her payoff.
\end{proof}

So far, most research on finding Nash equilibria in infinite games
has focused on computing \emph{some} Nash equilibrium \cite{ChatterjeeJM04}.
However, a game may have several Nash equilibria with different payoffs, and
one might not be interested in \emph{any} Nash equilibrium but in one whose
payoff fulfils certain requirements. For example, one might look for a
Nash equilibrium where certain players win almost surely while certain others
lose almost surely. This idea leads to the following decision problem, which
we call \NE:\footnote{In the definition of \NE, the ordering $\leq$ is applied
componentwise.}
\begin{quote}
Given an SMG $(\calG,v_0)$ and thresholds
$\vec{x},\vec{y}\in[0,1]^\Pi$, decide whether there exists a
Nash equilibrium of $(\calG,v_0)$ with payoff $\geq\vec{x}$ and
$\leq\vec{y}$.
\end{quote}
Of course, as a decision problem the problem only makes sense if the game and
the thresholds $\vec{x}$ and $\vec{y}$ are represented in a finite way. In the
following, we will therefore assume that the thresholds and all transition
probabilities are rational, and that all objectives are
$\omega$-regular.

Note that \NE puts no restriction on the type of strategies that realise the
equilibrium. It is natural to restrict the search space to equilibria that are
realised in pure, finite-state, stationary, or even positional strategies. Let
us call the corresponding decision problems \PureNE, \FinNE, \StatNE and
\PosNE, respectively.

In a recent paper \cite{UmmelsW09}, we studied \NE and its
variants in the context of \emph{simple} stochastic multiplayer games (SSMGs).
These are SMGs where each player's objective is to reach a certain set~$T$ of
\emph{terminal vertices}: $v\Delta=\{v\}$ for each $v\in T$.
In particular, such objectives are both B\"uchi and co-B\"uchi conditions. Our
main results on SSMGs can be summarised as follows:
\begin{itemize}
 \item \PureNE and \FinNE are undecidable;
 \item \StatNE is contained in \PSpace, but \NP- and \SqrtSum-hard;
 \item \PosNE is \NP-complete.
\end{itemize}
In fact, \PureNE and \FinNE are undecidable even if one restricts to instances
where the thresholds are binary, but distinct, or if one restricts to instances
where the thresholds coincide (but are not binary). Hence,
the question arises what happens if the thresholds are binary and coincide.
This question motivates the following \emph{qualitative} version of \NE, a
problem which we call \QualNE:
\begin{quote}
Given an SMG $(\calG,v_0)$ and $\vec{x}\in\{0,1\}^\Pi$, decide whether
$(\calG,v_0)$ has a Nash equilibrium with payoff $\vec{x}$.
\end{quote}

In this paper, we show that \QualNE, \StatNE and \PosNE are decidable for
games with arbitrary $\omega$-regular objectives, and analyse the
complexities of these problems depending on the type of the objectives.

\section{Stationary equilibria}

In this section, we analyse the complexity of the problems \PosNE and \StatNE.
Lower bounds for these problems follow from our results on SSMGs 
\cite{UmmelsW09}.

\begin{theorem}
\PosNE is \NP-complete for SMGs with B\"uchi, co-B\"uchi, parity, Rabin,
Streett, or Muller objectives.
\end{theorem}
\begin{proof}
Hardness was already proven in \cite{UmmelsW09}. To prove membership in \NP, we
give a nondeterministic polynomial-time algorithm for deciding \PosNE. On
input $\calG,v_0,\vec{x},\vec{y}$, the
algorithm simply guesses a positional strategy profile~$\vec{\sigma}$ (which is
basically a mapping $\bigcup_{i\in\Pi} V_i \to V$). Next, the algorithm computes
the payoff~$z_i$ of $\vec{\sigma}$ for each player~$i$ by
computing the probability of the event $\Win_i$ in the Markov chain
$(\calG^{\vec{\sigma}},v_0)$, which arises from~$\calG$ by fixing all
transitions according to $\vec{\sigma}$. Once each $z_i$ is computed, the
algorithm can easily check whether $x_i\leq z_i\leq y_i$.
To check whether~$\vec{\sigma}$ is a Nash equilibrium, the algorithm needs to
compute, for each player~$i$, the value $r_i$ of the MDP
$(\calG^{\vec{\sigma}_{-i}},v_0)$, which arises from $\calG$ by fixing all
transitions but the ones leaving vertices controlled by player~$i$
according to $\vec{\sigma}$ (and imposing the objective $\Win_i$). Clearly,
$\vec{\sigma}$ is a Nash equilibrium iff $r_i\leq z_i$ for each player~$i$.
Since we can compute the value of any MDP (and thus any Markov chain) with
one of the above objectives in polynomial time
\cite{Chatterjee07,Alfaro97}, all these checks can be carried out in polynomial
time.
\end{proof}

To prove the decidability of \StatNE, we appeal to results established
for the \emph{Existential Theory of the Reals}, \ExThR, the set of all
existential first-order sentences (over the appropriate signature) that hold in
$\frakR:=(\bbR,+,\cdot,0,1,\leq)$. The best known upper bound for the
complexity of the associated decision problem is \PSpace
\cite{Canny88}, which leads to the following theorem.

\begin{theorem}\label{thm:statne-pspace}
\StatNE is in \PSpace for SMGs with B\"uchi, co-B\"uchi, parity, Rabin,
Streett, or Muller objective.
\end{theorem}
\begin{proof}
Since $\PSpace=\NPSpace$, it suffices to provide a nondeterministic algorithm
with polynomial space requirements for deciding \StatNE. On input $\calG,v_0,
\vec{x},\vec{y}$, where \wlg $\calG$ is an SMG with Muller
objectives $\calF_i\in 2^V$, the algorithm starts by
guessing the \emph{support} $S\subseteq V\times V$ of a stationary strategy
profile~$\vec{\sigma}$ of $\calG$, \ie $S=\{(v,w)\in V\times V:\vec{\sigma}
(w\mid v)>0\}$. From the set~$S$ alone, by standard graph algorithms (see
\cite{Chatterjee07,Alfaro97}), one can
compute (in polynomial time) for each player~$i$ the following sets:
\begin{enumerate}
 \item the union~$F_i$ of all end components (\ie bottom SCCs) $C$ of the
Markov chain~$\calG^{\vec{\sigma}}$ that are winning
for player~$i$, \ie $C\in\calF_i$;

 \item the set~$R_i$ of vertices~$v$ such that
$\Prob_v^{\vec{\sigma}}(\Reach(F_i))>0$;

 \item the union~$T_i$ of all end components of the
MDP~$\calG^{\vec{\sigma}_{-i}}$ that are winning for player~$i$.
\end{enumerate}

After computing all these sets, the algorithm evaluates an existential
first-order sentence~$\psi$, which can be computed in polynomial time from
$\calG$, $v_0$, $\vec{x}$, $\vec{y}$, $(R_i)_{i\in\Pi}$, $(F_i)_{i\in\Pi}$ and
$(T_i)_{i\in\Pi}$ over $\frakR$ and returns the answer to this query.

It remains to describe a suitable sentence $\psi$.
Let $\vec{\alpha}=(\alpha_{vw})_{v,w\in V}$, $\vec{r}=(r^i_v)_{i\in\Pi,v\in V}$
and $\vec{z}=(z^i_v)_{i\in\Pi, v\in V}$ be three sets of variables, and let
$V_\ast=\bigcup_{i \in \Pi} V_i$ be the set of all non-stochastic vertices. The
formula
\begin{align*}
\phi(\vec{\alpha})&:=
\bigwedge_{v\in V_\ast} \bigg( \bigwedge_{w\in v\Delta} \alpha_{vw} \geq 0\
\wedge \bigwedge_{\mathmakebox[0.7cm][c]{w\in V\setminus v\Delta}}
\alpha_{vw}=0\wedge \sum_{w \in v\Delta} \alpha_{vw} = 1\bigg)\,\wedge \\
& \quad
\bigwedge_{\mathmakebox[0.7cm][c]{\substack{v\in V\setminus V_\ast \\ w\in V}}}
\alpha_{vw} = p_{vw}\wedge \bigwedge_{\mathmakebox[0.7cm][c]{(v,w)\in S}}
\alpha_{vw} > 0 \wedge
\bigwedge_{\mathmakebox[0.7cm][c]{(v,w)\not\in S}} \alpha_{vw}=0\,,
\end{align*}
where $p_{vw}$ is the unique number such that $(v,p_{vw},w)\in\Delta$,
states that the mapping~$\vec{\sigma}:V\to\calD(V)$ defined by
$\vec{\sigma}(w\mid v)=\alpha_{vw}$
constitutes a valid stationary strategy profile of $\calG$ whose support
is~$S$. Provided that $\phi(\vec{\alpha})$ holds in $\frakR$, the formula
\[\eta_i(\vec{\alpha},\vec{z}):=
\bigwedge_{v \in F_i} z^i_v= 1 \wedge
\bigwedge_{\mathmakebox[0.7cm][c]{v \in V\setminus R_i}} z^i_v= 0 \wedge
\bigwedge_{\mathmakebox[0.7cm][c]{v \in V\setminus F_i}} z^i_v =
\sum_{w \in v\Delta} \alpha_{vw} z^i_w\]
states that $z^i_v=\Prob_v^{\vec{\sigma}}(\Win_i)$ for each $v\in V$,
where $\vec{\sigma}$ is defined as above. This follows from a well-known
results about Markov chains, namely that the vector of the aforementioned
probabilities is the unique solution of the given system of equations. Finally,
the formula
\[\theta_i(\vec{\alpha},\vec{r}):=
\bigwedge_{v\in V} r^i_v\geq 0 \wedge
\bigwedge_{v\in T_i} r^i_v=1 \wedge
\bigwedge_{\substack{v\in V_i \\ w\in v\Delta}} r^i_v\geq r^i_w\wedge
\bigwedge_{\mathmakebox[0.7cm][c]{v\in V\setminus V_i}}
r^i_v=\sum_{w\in v\Delta}\alpha_{vw} r^i_w\]
states that $\vec{r}$ is a solution of the linear programme for computing
the maximal payoff that player~$i$ can achieve when playing against the
strategy profile~$\vec{\sigma}_{-i}$. In particular, the formula is fulfilled
if $r^i_v=\sup_{\tau}\Prob^{(\vec{\sigma}_{-i},\tau)}_v(\Reach(T_i))=
\sup_{\tau}\Prob^{(\vec{\sigma}_{-i},\tau)}_v(\Win_i)$ (where the latter
equality follows from
\cref{lemma:end-components,lemma:end-component-strategy}), and every other
solution is greater than this one (in each component).

The desired sentence $\psi$ is the existential closure of the conjunction of
$\phi$ and, for each player~$i$, the formulae $\eta_i$ and $\theta_i$ combined
with formulae stating that player~$i$ cannot improve her payoff and that the
expected payoff for player~$i$ lies in between the given thresholds:
\[\psi:=\exists\vec{\alpha}\,\exists\vec{r}\,\exists\vec{z}\,
\big(\phi(\vec{\alpha})\wedge\bigwedge_{i\in\Pi}(\eta_i(\vec{\alpha},\vec{z})
\wedge\theta_i(\vec{\alpha},\vec{r})\wedge r^i_{v_0}\leq z^i_{v_0}\wedge
x_i\leq z^i_{v_0}\leq y_i)\big).\]
It follows that $\psi$ holds in $\frakR$ iff $(\calG,v_0)$ has a
stationary Nash equilibrium with payoff at least $\vec{x}$ and
at most $\vec{y}$ whose support is~$S$. Consequently, the algorithm is correct.
\end{proof}

\section{Equilibria with a binary payoff}

In this section, we prove that \QualNE is decidable. We start by
characterising the existence of a Nash equilibrium with a binary payoff in any
game with prefix-independent objectives.

\subsection{Characterisation of existence}
\label{sect:nash-qual-characterisation}

For a subset $U\subseteq V$, we denote by $\Reach(U)$ the
set $V^*\cdot U\cdot V^\omega$; if $U=\{v\}$, we just write $\Reach(v)$ for
$\Reach(U)$. Finally, given an SMG~$\calG$ and a player~$i$, we denote by
$V_i^{>0}$ the set of all vertices $v\in V$ such that $\val_i^\calG(v)>0$.
The following lemma allows to infer the existence of a Nash equilibrium from 
the existence of a certain strategy profile. The proof uses so-called
\emph{threat strategies}
(also known as \emph{trigger strategies}), which are the basis of the
\emph{folk theorems} in the theory of repeated games
(\cf \cite[Chapter 8]{OsborneR94}).

\begin{lemma}\label{lemma:nash-threat}
Let $\vec{\sigma}$ be a pure strategy profile of $\calG$ such that, for each
player~$i$, $\Prob_{v_0}^{\vec{\sigma}}(\Win_i)=1$ or
$\Prob^{\vec{\sigma}}_{v_0}(\Reach(V_i^{>0}))=0$. Then there exists a
pure Nash equilibrium~$\vec{\sigma}^*$ with
$\Prob_{v_0}^{\vec{\sigma}}=\Prob_{v_0}^{\vec{\sigma}^*}$. If, additionally,
all winning conditions are $\omega$-regular and $\vec{\sigma}$ is
finite-state, then there exists a finite-state Nash
equilibrium~$\vec{\sigma}^*$ with $\Prob_{v_0}^{\vec{\sigma}}=
\Prob_{v_0}^{\vec{\sigma}^*}$. 
\end{lemma}
\begin{proof}
Consider the 2SG~$\calG_i=
(\{i,\Pi\setminus\{i\}\},V,V_i,\bigcup_{j\not=i} V_j,\Delta,\Win_i,V^\omega\setminus\Win_i)$ where
player~$i$ plays against the coalition~$\Pi\setminus\{i\}$ of all other
players.
Since the set $\Win_i$ is prefix-independent, there exists a globally
optimal pure strategy~$\tau_i$ for the coalition in this game. For each
player~$j\not=i$, this strategy induces a pure strategy~$\tau_{j,i}$ in
$\calG$.
To simplify notation, we also define $\tau_{i,i}$ to be an arbitrary
finite-state strategy of player~$i$ in $\calG$. Player~$i$'s
strategy~$\sigma^*_i$ in $\vec{\sigma}^*$ is defined as follows:
\[\sigma^*_i(xv)=\begin{cases} \sigma_i(xv) & \text{if
$\Prob^{\vec{\sigma}}_{v_0}(xv\cdot V^\omega)>0$,} \\
\tau_{i,j}(x_2v) & \text{otherwise,}
\end{cases}\]
where, in the latter case, $x=x_1x_2$ with $x_1$ being the longest prefix of
$xv$ such that $\Prob^{\vec{\sigma}}_{v_0}(x_1\cdot V^\omega)>0$ and $j\in\Pi$
being the player that has deviated from $\vec{\sigma}$, \ie $x_1$ ends in
$V_j$; if $x_1$ is empty or ends in a stochastic vertex, we set $j=i$.
Intuitively, $\sigma_i^*$ behaves like $\sigma_i$ as long as no other
player~$j$ deviates from playing $\sigma_j$, in which case $\sigma_i^*$
starts to behave like $\tau_{i,j}$.

If each $\Win_i$ is $\omega$-regular, then $\tau$ can be chosen to be a
finite-state profile. Consequently, each $\tau_{j,i}$ can be assumed to be
finite-state. If additionally $\vec{\sigma}$ is finite-state, it is easy to
see that the strategy profile~$\vec{\sigma}^*$, as defined above, is also
finite-state.

Note that $\Prob^{\vec{\sigma}^*}_{v_0}=\Prob^{\vec{\sigma}}_{v_0}$. We claim
that $\vec{\sigma}^*$ is a Nash
equilibrium of $(\calG,v_0)$. Let $\rho$ be any strategy of player~$i$ in
$\calG$; we need to show that $\Prob^{\vec{\sigma}^*_{-i},\rho}_{v_0}(\Win_i)
\leq\Prob^{\vec{\sigma}^*}_{v_0}(\Win_i)$.

Let us call a history $xvw\in V^*\cdot V_i\cdot V$ a
\emph{deviation history} if $\Prob_{v_0}^{\vec{\sigma}}(xv\cdot V^\omega)>0$,
but $\sigma_i(xv)\not=w$ and $\rho(w\mid xv)>0$; we denote the set of all
deviation histories by~$X$.

\begin{claim*}
$\Prob^{\vec{\sigma}^*_{-i},\rho}_{v_0}(B\setminus X\cdot V^\omega)\leq
\Prob^{\vec{\sigma}}_{v_0}(B)$ for every Borel set~$B$.
\end{claim*}
\begin{proof}
The claim is obviously true for the \emph{basic open sets} $B=w\cdot V^\omega$
(where $w\in V^*$) and thus also for finite, disjoint unions of such sets,
which are precisely the \emph{clopen sets} (\ie sets of the form $W\cdot
V^\omega$ for finite $W\subseteq V^*$). Since the class of clopen sets is
closed under complements and finite unions, by the \emph{monotone class
theorem} \cite{Halmos74},
the closure of the class of all clopen sets under
taking limits of chains contains the smallest $\sigma$-algebra containing all
clopen sets, which is just the Borel $\sigma$-algebra.
Hence, it suffices to show that whenever we are given measurable sets
$A_1,A_2,\ldots\subseteq V^\omega$ with $A_1\subseteq A_2
\subseteq\dots$ or $A_1\supseteq A_2\supseteq\dots$ such that the claim
holds for each $A_n$, then the claim also holds for $\lim_n A_n$, where
$\lim_n A_n=\bigcup_{n\in\bbN} A_n$ or $\lim_n A_n=
\bigcap_{n\in\bbN} A_n$, respectively. 
So assume that $A_1,A_2,\dots\subseteq V^\omega$ is a chain such that
$\Prob^{\vec{\sigma}^*_{-i},\rho}_{v_0}(A_n\setminus X\cdot V^\omega)
\leq\Prob^{\vec{\sigma}}_{v_0}(A_n)$ for each $n\in\bbN$. Clearly,
$(\lim_n A_n)\setminus X\cdot V^\omega=\lim_n (A_n\setminus X\cdot V^\omega)$.
Moreover, since measures are continuous from above and below:
\begin{align*}
 & \Prob^{\vec{\sigma}^*_{-i},\rho}_{v_0}(\lim_n
 (A_n\setminus X\cdot V^\omega)) \\
=\; & \lim_n \Prob^{\vec{\sigma}^*_{-i},\rho}_{v_0}(A_n\setminus X\cdot
 V^\omega) \\
\leq\; & \lim_n \Prob^{\vec{\sigma}}_{v_0}(A_n) \\
=\; & \Prob^{\vec{\sigma}}_{v_0}(\lim_n A_n)\,.\qedhere
\end{align*}
\end{proof}

As usual in probability theory, if $P$ is a probability measure and $A$
and $B$ are measurable sets such that $P(B)>0$, then we denote by
$P(A\mid B)$ the \emph{conditional probability of $A$ given $B$}, defined
by $P(A\mid B)=\frac{P(A\cap B)}{P(B)}$.

\begin{claim*}
$\Prob^{\vec{\sigma}^*_{-i},\rho}_{v_0}(\Win_i\mid xvw\cdot V^\omega)
\leq\val_i^\calG(w)$ for every $xvw\in X$.
\end{claim*}
\begin{proof}
By the definition of the strategies $\tau_{j,i}$, we have that
$\Prob^{(\tau_{j,i})_{j\not=i},\rho}_{v}(\Win_i)\leq \val_i^\calG(v)$ for
every vertex $v\in V$ and \emph{every} strategy $\rho$ of player~$i$.
On the other hand, if $xvw$ is a deviation
history, then for each player~$j$ the residual strategy
$\sigma^*_j[xv]$ is equal to $\tau_{j,i}$ on histories that start
in $w$. Hence, by \cref{lemma:residual-strategies}, and since
the set $\Win_i$ is prefix-independent, we get:
\begin{align*}
& \Prob^{\vec{\sigma}^*_{-i},\rho}_{v_0}(\Win_i\mid xvw\cdot V^\omega) \\
=\; & \Prob^{\vec{\sigma}^*_{-i},\rho}_{v_0}(\Win_i\cap\>xvw\cdot V^\omega)\div
\Prob^{\vec{\sigma}^*_{-i},\rho}_{v_0}(xvw\cdot V^\omega) \\
=\; & \Prob^{\vec{\sigma}^*_{-i}[xv],\rho[xv]}_{w}(\Win_i) \\
=\; & \Prob^{(\tau_{j,i})_{j\not=i},\rho[xv]}_{w}(\Win_i) \\
\leq\; & \val_i^\calG(w)\qedhere
\end{align*}
\end{proof}

Using the previous two claims, we can prove that
$\Prob^{\vec{\sigma}^*_{-i},\rho}_{v_0}(\Win_i)\leq
\Prob^{\vec{\sigma}}_{v_0}(\Win_i)=\Prob^{\vec{\sigma}^*}_{v_0}(\Win_i)$ as 
follows:
\begin{align*}
& \Prob^{\vec{\sigma}^*_{-i},\rho}_{v_0}(\Win_i) \\
=\; &\Prob^{\vec{\sigma}^*_{-i},\rho}_{v_0}(\Win_i\setminus\> X\cdot V^\omega)
+\sum_{\mathmakebox[0.5cm][c]{xvw\in X}}
\Prob^{\vec{\sigma}^*_{-i},\rho}_{v_0}(\Win_i\cap\>xvw\cdot V^\omega) \\
\leq\; &  \Prob^{\vec{\sigma}}_{v_0}(\Win_i)
+\sum_{\mathmakebox[0.5cm][c]{xvw\in X}}
\Prob^{\vec{\sigma}^*_{-i},\rho}_{v_0}(\Win_i\cap\>xvw\cdot V^\omega) \\
=\; & \Prob^{\vec{\sigma}}_{v_0}(\Win_i)
+\sum_{\mathmakebox[0.5cm][c]{xvw\in X}}
\Prob^{\vec{\sigma}^*_{-i},\rho}_{v_0}(\Win_i\mid xvw\cdot V^\omega)
\cdot\Prob^{\vec{\sigma}^*_{-i},\rho}_{v_0}(xvw\cdot V^\omega) \\
\leq\; & \Prob^{\vec{\sigma}}_{v_0}(\Win_i)
+\sum_{\mathmakebox[0.5cm][c]{xvw\in X}} \val_i^\calG(w)
\cdot\Prob^{\vec{\sigma}^*_{-i},\rho}_{v_0}(xvw\cdot V^\omega) \\
\leq\; & \Prob^{\vec{\sigma}}_{v_0}(\Win_i)
+\sum_{\mathmakebox[0.5cm][c]{xvw\in X}} \val_i^\calG(v)
\cdot\Prob^{\vec{\sigma}^*_{-i},\rho}_{v_0}(xvw\cdot V^\omega) \\
=\; & \Prob^{\vec{\sigma}}_{v_0}(\Win_i)\, ,
\end{align*}
where the last equality follows from $\Prob^{\vec{\sigma}}_{v_0}
(\Reach(V_i^{>0}))=0$, which implies that $\val_i^\calG(v)=0$ for each $v\in V$
such that $\Prob^{\vec{\sigma}}_{v_0}(\Reach(v))>0$.
\end{proof}

\noindent
Finally, we can state the main result of this section.

\begin{proposition}\label{prop:qual-nash}
Let $(\calG,v_0)$ be any SMG with prefix-independent winning
conditions, and let $\vec{x}\in\{0,1\}^\Pi$. Then the following statements
are equivalent:
\begin{enumerate}
 \item There exists a Nash equilibrium with payoff~$\vec{x}$;
 \item There exists a strategy profile~$\vec{\sigma}$ with
payoff~$\vec{x}$ such that\linebreak $\Prob^{\vec{\sigma}}_{v_0}(\Reach(V_i^{>0}))=0$
for each player~$i$ with $x_i=0$;
 \item There exists a pure strategy profile~$\vec{\sigma}$ with
payoff~$\vec{x}$ such that\linebreak $\Prob^{\vec{\sigma}}_{v_0}
(\Reach(V_i^{>0}))=0$ for each player~$i$ with $x_i=0$;
 \item There exists a pure Nash equilibrium with payoff~$\vec{x}$.
\end{enumerate}
If additionally all winning conditions are $\omega$-regular, then any of the
above statements is equivalent to each of the following statements:
\begin{enumerate}\setcounter{enumi}{4}
 \item There exists a finite-state strategy profile~$\vec{\sigma}$ with
payoff~$\vec{x}$ such that\linebreak $\Prob^{\vec{\sigma}}_{v_0}
(\Reach(V_i^{>0}))=0$ for each player~$i$ with $x_i=0$;
 \item There exists a finite-state Nash equilibrium with payoff~$\vec{x}$.
\end{enumerate}
\end{proposition}
\begin{proof}
$(1.\Rightarrow 2.)$ Let $\vec{\sigma}$ be a Nash equilibrium with
payoff~$\vec{x}$. We claim that $\vec{\sigma}$ is already the strategy profile
we are looking for: $\Prob^{\vec{\sigma}}_{v_0}(\Reach(V_i^{>0}))=0$ for each
player~$i$ with $x_i=0$.
Towards a contradiction, assume that $\Prob^{\vec{\sigma}}_{v_0}
(\Reach(V_i^{>0}))>0$ for some player~$i$ with $x_i=0$. Since $V$ is finite,
there exists a vertex $v\in V_i^{>0}$ and a history $x$ such that
$\Prob^{\vec{\sigma}}_{v_0}(xv\cdot V^\omega)>0$. Let $\tau$ be
an optimal strategy for player~$i$ in the game $(\calG,v)$, and consider
her strategy~$\sigma'$ defined by
\[\sigma'(yw)=\begin{cases}\sigma(yw) & \text{if $xv\npreceq yw$,} \\
\tau(y'w) & \text{otherwise,}
\end{cases}\]
where, in the latter case, $y=x y'$. Clearly, $\Prob^{\vec{\sigma}}_{v_0}
(xv\cdot V^\omega)=\Prob^{(\vec{\sigma}_{-i},\sigma')}_{v_0}
(xv\cdot V^\omega)$. Moreover, $\Prob^{\vec{\sigma}}_{v_0}
(\Win_i\setminus xv\cdot V^\omega)=\Prob^{(\vec{\sigma}_{-i},\sigma')}_{v_0}
(\Win_i\setminus xv\cdot V^\omega)$: this follows from
\cref{lemma:equivalent-profiles} by taking $X=V^*\setminus xv\cdot V^*$.
Using \cref{lemma:residual-strategies}, we can infer that
$\Prob_{v_0}^{(\vec{\sigma}_{-i},\sigma')}(\Win_i)>0$ as follows:
\begin{align*}
& \Prob_{v_0}^{(\vec{\sigma}_{-i},\sigma')}(\Win_i) \\
=\; & \Prob_{v_0}^{(\vec{\sigma}_{-i},\sigma')}(\Win_i\cap xv\cdot V^\omega)
 +\Prob_{v_0}^{(\vec{\sigma}_{-i},\sigma')}(\Win_i\setminus xv\cdot V^\omega)\\
=\; & \Prob_{v_0}^{(\vec{\sigma}_{-i},\sigma')}(xv\cdot V^\omega)
 \cdot\Prob_{v}^{(\vec{\sigma}_{-i},\sigma')[x]}(\Win_i)
 +\Prob_{v_0}^{\vec{\sigma}}(\Win_i\setminus xv\cdot V^\omega) \\
=\; & \Prob_{v_0}^{\vec{\sigma}}(xv\cdot V^\omega)
 \cdot\Prob_{v}^{(\vec{\sigma}_{-i}[x],\tau)}(\Win_i)
 +\Prob_{v_0}^{\vec{\sigma}}(\Win_i\setminus xv\cdot V^\omega) \\
\geq\; & \Prob_{v_0}^{\vec{\sigma}}(xv\cdot V^\omega)\cdot\val_i^\calG(v)
 +\Prob_{v_0}^{\vec{\sigma}}(\Win_i\setminus xv\cdot V^\omega) \\
>\; & 0
\end{align*}
Hence, player~$i$ can improve her payoff by playing $\sigma'$ instead of 
$\sigma_i$, a contradiction to the fact that $\vec{\sigma}$ is a Nash
equilibrium.

$(2.\Rightarrow 3.)$ Let $\vec{\sigma}$ be a strategy profile of $(\calG,v_0)$
with payoff~$\vec{x}$ such that $\Prob^{\vec{\sigma}}_{v_0}
(\Reach(V_i^{>0}))=0$ for each player~$i$ with $x_i=0$. Consider the MDP
$\calM$ that is obtained from $\calG$ by removing all vertices~$v\in V$ such
that $v\in V_i^{>0}$ for some player~$i$ with $x_i=0$, merging all players
into one, and imposing the objective
\[\Win=\bigcap\limits_{\substack{i\in\Pi \\ x_i=1}} \Win_i\cap
  \bigcap\limits_{\substack{i\in\Pi \\ x_i=0}} V^\omega\setminus\Win_i\,.\]
The MDP~$\calM$ is well-defined since its domain is a subarena of $\calG$.
Moreover, the value $\val^{\calM}(v_0)$ of $\calM$ is equal to~1
because the strategy profile~$\vec{\sigma}$ induces a
strategy~$\sigma$ in $\calM$ satisfying $\Prob_{v_0}^\sigma(\Win)=1$.
Since each $\Win_i$ is prefix-independent, so is the set $\Win$.
Hence, there exists a pure, optimal strategy~$\tau$ in $(\calM,v_0)$. Since
the value is~1, we have $\Prob_{v_0}^\tau(\Win)=1$, and $\tau$ induces
a pure strategy profile of $\calG$ with the desired properties.

$(3.\Rightarrow 4.)$ Let $\vec{\sigma}$ be a pure strategy
profile of $(\calG,v_0)$ with payoff~$\vec{x}$ such that
$\Prob^{\vec{\sigma}}_{v_0}(\Reach(V_i^{>0}))=0$ for each player~$i$ with
$x_i=0$. By \cref{lemma:nash-threat}, there exists a pure Nash
equilibrium~$\vec{\sigma}^*$ of $(\calG,v_0)$ with
$\Prob^{\vec{\sigma}}_{v_0}=\Prob^{\vec{\sigma}^*}_{v_0}$. In particular,
$\vec{\sigma}^*$ has payoff~$\vec{x}$.

$(4.\Rightarrow 1.)$ Trivial.

Under the additional assumption that all winning conditions are
$\omega$-regular, the implications $(2.\Rightarrow 5.)$ and
$(5.\Rightarrow 6.)$ are proven analogously; the implication
$(6.\Rightarrow 1.)$ is trivial.
%
%
%
\end{proof}

As an immediate consequence of \cref{prop:qual-nash}, we can conclude that
finite-state strategies are as powerful as arbitrary mixed strategies as far as
the existence of a Nash equilibrium with a binary payoff in SMGs with
$\omega$-regular objectives is concerned. (This is not true for Nash equilibria
with a non-binary payoff \cite{Ummels08}.)

\begin{corollary}\label{cor:nash-finite-state}
Let $(\calG,v_0)$ be any SMG with $\omega$-regular objectives, and
let $x\in\{0,1\}^\Pi$. There exists a Nash equilibrium of $(\calG,v_0)$ with
payoff~$\vec{x}$ iff there exists a finite-state Nash equilibrium of
$(\calG,v_0)$ with payoff~$\vec{x}$.
\end{corollary}
\begin{proof}
The claim follows from \cref{prop:qual-nash} and the fact that every SMG with
$\omega$-regular objectives can be reduced to one with
prefix-independent $\omega$-regular (\eg parity) objectives.
\end{proof}

\subsection{Computational complexity}

We can now describe an algorithm for deciding \QualNE for games with Muller
objectives. The algorithm relies on the characterisation we gave in
\cref{prop:qual-nash}, which allows to reduce the problem to a problem
about a certain MDP.

Formally, given an SMG
$\calG=(\Pi,V,(V_i)_{i\in\Pi},\Delta,(\calF_i)_{i\in\Pi})$ with Muller
objectives
$\calF_i\subseteq 2^V$, and a binary payoff $\vec{x}\in\{0,1\}^\Pi$, we
define the Markov decision process~$\calG(\vec{x})$ as follows: Let
$Z\subseteq V$ be the set of all $v$ such that $\val_i^\calG(v)=0$ for each
player~$i$ with $x_i=0$; the set of vertices of $\calG(\vec{x})$ is precisely
the set~$Z$, with the set of vertices controlled by player~0
being $Z_0:=Z\cap\bigcup_{i\in\Pi} V_i$. (If $Z=\emptyset$, we define
$\calG(\vec{x})$ to be a trivial MDP with the empty set as its objective.)
The transition relation
of~$\calG(\vec{x})$ is the restriction of $\Delta$ to transitions between
$Z$-states. Note that the transition relation of $\calG(\vec{x})$ is
well-defined since $Z$ is a subarena of $\calG$. We say that a
subset $U\subseteq V$ has payoff~$\vec{x}$ if $U\in\calF_i$ for each
player~$i$ with $x_i=1$ and $U\not\in\calF_i$ for each player~$i$ with
$x_i=0$. The objective of $\calG(\vec{x})$ is $\Reach(T)$ where
$T\subseteq Z$ is the union of all end components $U\subseteq Z$ that have
payoff~$\vec{x}$.

\begin{lemma}\label{lemma:nash-muller}
Let $(\calG,v_0)$ be any SMG with Muller objectives, and let $\vec{x}
\in\{0,1\}^\Pi$. Then $(\calG,v_0)$ has a Nash equilibrium with
payoff~$\vec{x}$ iff $\val^{\calG(\vec{x})}(v_0)=1$.
\end{lemma}
\begin{proof}
$(\Rightarrow)$ Assume that $(\calG,v_0)$ has a Nash equilibrium with
payoff~$\vec{x}$. By \cref{prop:qual-nash}, this implies that there exists a
strategy profile~$\vec{\sigma}$ of $(\calG,v_0)$ with payoff~$\vec{x}$ such
that $\Prob_{v_0}^{\vec{\sigma}}(\Reach(V\setminus Z))=0$.
We claim that $\Prob^{\vec{\sigma}}_{v_0}(\Reach(T))=1$.
Otherwise, by \cref{lemma:end-components}, there would exist an end component
$C\subseteq Z$ such that $C\not\in\calF_i$ for some player~$i$ with
$x_i=1$ or $C\in\calF_i$ for some some player~$i$ with $x_i=0$, and
$\Prob^{\vec{\sigma}}_{v_0}(\{\alpha\in V^\omega:\Inf(\alpha)=C\})>0$. But 
then, $\vec{\sigma}$ cannot have payoff~$\vec{x}$, a contradiction.
Now, since $\Prob_{v_0}^{\vec{\sigma}}(\Reach(V\setminus Z))=0$, $\vec{\sigma}$
induces a strategy~$\sigma$ in $\calG(\vec{x})$ such that
$\Prob^\sigma_{v_0}(B)=\Prob^{\vec{\sigma}}_{v_0}(B)$ for every Borel set
$B\subseteq Z^\omega$. In particular, $\Prob^\sigma_{v_0}(\Reach(T))=1$ and
hence $\val^{\calG(\vec{x})}(v_0)=1$.

$(\Leftarrow)$ Assume that $\val^{\calG(\vec{x})}(v_0)=1$ (in particular,
$v_0\in Z$), and let $\sigma$ be an optimal strategy in $(\calG(\vec{x}),v_0)$.
From $\sigma$, using \cref{lemma:end-component-strategy}, we can devise a
strategy~$\sigma'$ such that $\Prob^{\sigma'}_{v_0}
(\{\alpha\in V^\omega:\text{$\Inf(\alpha)$ has payoff~$\vec{x}$}\})=1$.
Finally, $\sigma'$ can can be extended to
a strategy profile $\vec{\sigma}$ of $\calG$ with payoff $\vec{x}$ such
that $\Prob_{v_0}^{\vec{\sigma}}(\Reach(V\setminus Z))=0$. By 
\cref{prop:qual-nash}, this implies that $(\calG,v_0)$ has a Nash equilibrium
with payoff~$\vec{x}$.
\end{proof}

Since the value of an MDP with reachability objectives can be computed in
polynomial time (via linear programming,
\cf \cite{Puterman94}), the difficult part lies in computing the MDP
$\calG(\vec{x})$ from $\calG$ and $\vec{x}$ (\ie its domain~$Z$ and
the target set~$T$).

\begin{theorem}\label{thm:nash-muller}
\QualNE is in \PSpace for games with Muller objectives.
\end{theorem}
\begin{proof}
Since $\PSpace=\NPSpace$, it suffices to give a nondeterministic algorithm
with polynomial space requirements. On input~$\calG,v_0,\vec{x}$, the
algorithm starts by computing for each player~$i$ with $x_i=0$ the set of
vertices~$v$ with $\val_i^\calG(v)=0$, which can be done in polynomial space
(see \cref{table:summary-results}). The intersection of these sets is the
domain~$Z$ of the Markov decision process~$\calG(\vec{x})$. If $v_0$ is not
contained in this intersection, the algorithm immediately rejects. Otherwise,
the algorithm proceeds by guessing a set~$T'\subseteq Z$ and for each
$v\in T'$ a set $U_v\subseteq Z$ with $v\in U_v$. If, for each $v\in T'$, the
set $U_v$ is an end component with payoff $\vec{x}$, the algorithm proceeds by
computing (in polynomial time) the value $\val^{\calG(\vec{x})}(v_0)$ of the
MDP~$\calG(\vec{x})$ with $T'$ substituted for $T$ and accepts if the value
is~1. In all other cases, the algorithm rejects.

The correctness of the algorithm follows from \cref{lemma:nash-muller} and
the fact that $\Prob^{\sigma}_{v_0}(\Reach(T'))\leq\Prob^{\sigma}_{v_0}
(\Reach(T))$ for any strategy $\sigma$ in $\calG(\vec{x})$ and any subset
$T'\subseteq T$.
\end{proof}

Since any SMG with $\omega$-regular can effectively be reduced to one with
Muller objectives, \cref{thm:nash-muller} implies the decidability
of \QualNE for games with arbitrary $\omega$-regular objectives (\eg given by
\SOneS formulae).
Regarding games with Muller objectives, a matching \PSpace-hardness result
appeared in \cite{HunterD05}, where it was shown that the qualitative decision
problem for 2SGs with Muller objectives is \PSpace-hard, even for games without
stochastic vertices. However, this result relies on the use of arbitrary
colourings.

To solve \QualNE for games with Streett objectives, we will make use
of the following procedure $\StreettEC(U)$, which computes for a game $\calG$
with Streett objectives $\Omega_i$, $i\in\Pi$, and a binary payoff~$\vec{x}\in
\{0,1\}^\Pi$ the union of all end components with payoff~$\vec{x}$ that are
contained in $U\subseteq V$.
\begin{tabbing}
\hspace*{1em}\=\hspace{1em}\=\hspace{1em}\=\hspace{1em}\=
\hspace{1em}\=\hspace{1em}\= \kill

\+\textbf{procedure} $\StreettEC(U)$ \\
$Z:=\emptyset$ \\
Compute (in polynomial time) all end components of $\calG$ maximal in $U$ \\
\+\textbf{for each} such end component $C$ \textbf{do} \\
$S:=\{i\in\Pi:\text{$x_i=1$ and ex.\ $(F,G)\in\Omega_i$ s.th.\ 
  $C\cap F\not=\emptyset$ and $C\cap G=\emptyset$}\}$ \\
$R:=\{i\in\Pi:\text{$x_i=0$ and ($C\cap F=\emptyset$ or
  $C\cap G\not=\emptyset$) for all $(F,G)\in\Omega_i$}\}$ \\
\+\textbf{if} $S=R=\emptyset$ \textbf{then} \\
$Z:=Z\cup C$ \-\\
\+\textbf{else if} $S\not=\emptyset$ \textbf{then} \\
$Y:=C\cap\bigcap_{i\in S}\bigcap_{(F,G)\in\Omega_i,C\cap G=\emptyset}
  C\setminus F$ \\
$Z:=Z\cup\StreettEC(Y)$ \-\\
\+\textbf{else if} $R\not=\emptyset$ and $C\cap F\not=\emptyset$ for
all $(F,G)\in\Omega _i$, $i\in R$ \textbf{then} \\
$Y:=C\cap\bigcap_{i\in R}\bigcap_{(F,G)\in\Omega_i} C\setminus G$ \\
$Z:=Z\cup\StreettEC(Y)$ \-\\
\textbf{end if} \-\\
\textbf{end for} \\
\textbf{return} $Z$ \-\\
\textbf{end procedure}
\end{tabbing}

Note that on input~$U$, $\StreettEC$ calls itself at most $|U|$ times; hence,
the procedure runs in polynomial time.
Moreover, we can obtain a polynomial-time procedure $\RabinEC$ that computes
the same output for games with Rabin objectives~$\Omega_i$ by switching
$x_i=0$ and $x_i=1$ in the definitions of~$S$ and~$R$.

\begin{theorem}\label{thm:nash-streett}
\QualNE is \NP-complete for games with Streett objectives.
\end{theorem}
\begin{proof}
Hardness was already proven in \cite{Ummels08}. To prove membership in \NP,
we describe a nondeterministic, polynomial-time algorithm:
On input~$\calG,v_0,\vec{x}$, the algorithm starts by guessing a subarena~$Z'
\subseteq V$ and, for each player~$i$ with $x_i=0$, a positional strategy
$\tau_i$ of the coalition~$\Pi\setminus\{i\}$ in the 2SG~$\calG_i$, as defined
in the proof of \cref{lemma:nash-threat}.
In the next step, the algorithm checks (in polynomial time)
whether $\val^{\tau_i}(v)=1$ for each vertex $v\in Z'$ and
each player~$i$ with $x_i=0$. If not, the algorithm rejects immediately.
Otherwise, the algorithm proceeds by calling the procedure $\StreettEC$ to
determine the
union~$T'$ of all end components with payoff~$\vec{x}$ that are contained in
$S'$. Finally, the algorithm computes (in polynomial time) the
value~$\val^{\calG(\vec{x})}(v_0)$ of the MDP~$\calG(\vec{x})$ with $Z'$
substituted for $Z$ and $T'$ substituted for $T$. If this value is~1, the
algorithm accepts; otherwise, it rejects.

It remains to show that the algorithm is correct: On the one hand, if
$(\calG,v_0)$ has a Nash equilibrium with payoff~$\vec{x}$, then the run
of the algorithm where it guesses $Z'=Z$ and globally optimal positional
strategies $\tau_i$ (which exist since in the games $\calG_i$ the coalition has
a Rabin objective) will be accepting since then $T'=T$
and, by \cref{lemma:nash-muller}, $\val^{\calG(\vec{x})}(v_0)=1$.
On the other hand, in any accepting run of the algorithm we have
$Z'\subseteq Z$ and $T'\subseteq T$, and the value that the algorithm
computes cannot be higher than $\val^{\calG(\vec{x})}(v_0)$; hence,
$\val^{\calG(\vec{x})}(v_0)=1$, and \cref{lemma:nash-muller} guarantees the
existence of a Nash equilibrium with payoff~$\vec{x}$.
\end{proof}

\begin{theorem}\label{thm:nash-rabin}
\QualNE is \coNP-complete for games with Rabin objectives.
\end{theorem}
\begin{proof}
Hardness is proven by a slight modification of the reduction for
demonstrating \NP-hardness of \QualNE for games with Streett objectives (see
the appendix). To show membership in \coNP, we describe a nondeterministic,
polynomial-time algorithm for the complement of \QualNE.
On input~$\calG,v_0,\vec{x}$, the algorithm starts by guessing a subarena
$Z'\subseteq V$ and, for each player~$i$ with $x_i=0$, a positional strategy
$\sigma_i$ of player~$i$ in $\calG$. In the next step,
the algorithm checks whether for each vertex~$v\in Z'$ there exists some
player~$i$ with $x_i=0$ and $\val^{\sigma_i}(v)>0$. If not, the algorithm
rejects immediately. Otherwise, the algorithm proceeds by calling the procedure
$\RabinEC$ to determine the union~$T'$ of all end components with
payoff~$\vec{x}$ that are contained in $V\setminus Z'$. Finally, the algorithm
computes (in polynomial time) the value~$\val^{\calG(\vec{x})}(v_0)$ of the
MDP~$\calG(\vec{x})$ with $V\setminus Z'$ substituted for $Z$ and $T'$
substituted for $T$. If this value is not~1, the algorithm accepts; otherwise,
it rejects.

The correctness of the algorithm is proven in a similar fashion as in the
proof of the previous theorem.
\end{proof}

Since any parity condition can be turned into both a Streett and a Rabin
condition where the number of pairs is linear in the number of priorities,
we can immediately infer from \cref{thm:nash-streett,thm:nash-rabin}
that \QualNE is in $\NP\cap\coNP$ for games with parity objectives.

\begin{corollary}\label{cor:nash-parity-np-conp}
\QualNE is in $\NP\cap\coNP$ for games with parity objectives.
\end{corollary}

It is a major open problem whether the qualitative (or even the quantitative)
decision problem for 2SGs with parity objectives is in \PTime. This
would imply that \QualNE is decidable in polynomial time for games with parity
objectives since this would allow us to compute the domain of the
MDP~$\calG(\vec{x})$ in polynomial time. For each $d\in\bbN$, a class of games
where the qualitative decision problem is provably in \PTime is the class of
all 2SGs with parity objectives that uses at most $d$ priorities
\cite{ChatterjeeJH03}. For $d=2$, this class includes all 2SGs with a B\"uchi
or a co-B\"uchi objective (for player~$0$). Hence, we have the following
theorem.

\begin{theorem}
For each $d\in\bbN$, \QualNE is in \PTime for games with parity winning
conditions that use at most $d$ priorities. In particular, \QualNE is in
\PTime for games with (co-)B\"uchi objectives.
\end{theorem}

\section{Conclusion}

We have analysed the complexity of deciding whether a stochastic
multiplayer game with $\omega$-regular objectives has a Nash equilibrium whose
payoff falls into a certain interval. Specifically, we have isolated several
decidable restrictions of the general problem that have a manageable complexity
(\PSpace at most). For instance, the complexity of the qualitative variant of
\NE is usually not higher than for the corresponding problem for two-player
zero-sum games.

Apart from settling the complexity of \NE (where arbitrary mixed strategies are
allowed), two directions for future work come to mind: First, one could study 
other restrictions of \NE that might be decidable. For example, it seems
plausible that the restriction of \NE to games with two players is decidable.
Second, it seems interesting to see whether our decidability results can be
extended to more general models of games, \eg \emph{concurrent games} or games
with infinitely many states like \emph{pushdown games}.

\bibliographystyle{mgiabbrv}
\bibliography{../../biblio/all.bib}

\section*{Appendix}

\begin{theorem*}
\QualNE is \coNP-hard for games with Rabin objectives.
\end{theorem*}
\begin{proof}
The proof is a variant of the proof for \NP-hardness of the problem of deciding
whether player~$0$ has a winning strategy in a two-player zero-sum game with a
Rabin objective \cite{EmersonJ88} and by a reduction from the
unsatisfiability problem for Boolean formulae.

Given a Boolean formula $\phi$ in
conjunctive normal form, we construct a two-player SMG $\calG_\phi$ without
any stochastic vertex as follows: For each clause $C$ the game $\calG_\phi$ has
a vertex $C$, which is controlled by player~0, and for each literal $X$
or $\neg X$ occurring in $\phi$ there is a vertex $X$ or $\neg X$,
respectively, which is controlled by player~1. There are edges from a clause
to each literal that occurs in this clause, and from a literal to every clause
occurring in $\phi$.
Player~1's objective is given by the single Rabin pair
$(V,\emptyset)$, \ie she always wins, whereas player~0's objective
consists of all Rabin pairs of the form $(\{X\},\{\neg X\})$ and
$(\{\neg X\},\{X\})$.

Obviously, $\calG_\phi$ can be constructed from $\phi$ in polynomial time. We
claim that $\phi$ is unsatisfiable if and only if $(\calG_\phi,C)$
has a Nash equilibrium with payoff~$(0,1)$ (where $C$ is an arbitrary clause).

$(\Rightarrow)$ Assume that $\phi$ is not satisfiable. We claim that player~1
has a strategy $\tau$ to ensure that player~0's objective is
violated. Consequently, for any strategy~$\sigma$ of player~0, the strategy
profile $(\sigma,\tau)$ is a Nash equilibrium with payoff~$(0,1)$. Otherwise,
let $\sigma$ be a positional optimal strategy for player~0. By determinacy,
this strategy ensures that player~0's objective is satisfied. But a
positional strategy $\sigma$ of player~1 chooses for each clause a literal
contained in this clause. Since $\phi$ is unsatisfiable, there must exist a
variable $X$ and clauses $C_1$ and $C_2$ such that $\sigma(C_1)=X$ and
$\sigma(C_2)=\neg X$. Player~2's counter strategy is to play from $X$ to
$C_2$ and from any other literal to $C_1$. So the strategy $\sigma$ is not
optimal, a contradiction.

$(\Leftarrow)$ Assume that $\phi$ is satisfiable. Consider player~1's
positional strategy $\sigma$ of playing from a clause to a literal that
satisfies this clause. This ensures that for each variable $X$ at most
one of the literals $X$ or $\neg X$ is visited infinitely often.
The value of $\sigma$ from any vertex is~1; hence, there can be no Nash
equilibrium with payoff~$(0,1)$.
\end{proof}

\end{document}